\documentclass[preprint]{elsarticle}
\usepackage[T1]{fontenc}
\usepackage[english]{babel}

\usepackage{amsmath, amssymb}
\usepackage{graphicx}
\usepackage{tikz}
\usetikzlibrary{arrows,automata,shapes}
\usepackage{xspace, comment}

\newcommand{\N}{\ensuremath{\mathbb{N}}\xspace}
\newcommand{\A}{\ensuremath{\mathcal{A}}\xspace}
\newcommand{\aStar}{\ensuremath{\Sigma^{*}}\xspace}
\newcommand{\aN}{\ensuremath{\Sigma^{\omega}}\xspace}

\newcommand{\powerset}[1]{\ensuremath{\mathcal{P}\left( #1 \right)}\xspace}
\renewcommand{\L}{\ensuremath{\mathcal{L}}\xspace}

\newcommand{\calF}{\ensuremath{\mathcal{F}}\xspace}
\newcommand{\calL}{\ensuremath{\mathcal{L}}\xspace}
\newcommand{\Nlanguage}{\ensuremath{\omega}-language\xspace}
\newcommand{\Nrational}{\ensuremath{\omega}-rational\xspace}
\newcommand{\aut}{\ensuremath{\mathcal{A}}\xspace}
\newcommand{\set}[1]{\ensuremath{\left\{ #1  \right\}}\xspace}
\newcommand{\card}[1]{\left| #1 \right|\xspace}
\newcommand{\Lcond}[2]{\ensuremath{\mathcal{L}^{#1}_{ #2}}\xspace} 
\newcommand{\bool}[1]{\ensuremath{\mathcal{B} \left( #1 \right)}\xspace}
\newcommand{\ie}{\emph{i.e.}\@\xspace}
\newcommand{\etc}{\emph{etc.}\@\xspace}
\newcommand{\wrt}{\emph{w.r.t.}\@\xspace}

\newcommand{\cL}{\mathcal{L}}
\newcommand{\cA}{\ensuremath{\mathcal{A}}\xspace}
\newcommand{\FV}[1]{\ensuremath{\mathrm{FV}\left( #1 \right)}\xspace}

\newcommand{\FA}{\ensuremath{\mathrm{FA}}\xspace}
\newcommand{\CFA}{\ensuremath{\mathrm{CFA}}\xspace}
\newcommand{\DFA}{\ensuremath{\mathrm{DFA}}\xspace}
\newcommand{\CDFA}{\ensuremath{\mathrm{CDFA}}\xspace}

\newcommand{\run}{\ensuremath{\mathrm{run}}}
\newcommand{\fin}{\ensuremath{\mathrm{fin}}}
\newcommand{\ninf}{\ensuremath{\mathrm{ninf}}}

\newcommand{\F}{\ensuremath{\mathsf{F}}\xspace}
\newcommand{\FR}{\ensuremath{\mathsf{F^{R}}}\xspace}
\newcommand{\G}{\ensuremath{\mathsf{G}}\xspace}
\newcommand{\GR}{\ensuremath{\mathsf{G^{R}}}\xspace}
\newcommand{\Fs}{\ensuremath{\mathsf{F_{\sigma}}}\xspace}
\newcommand{\FsR}{\ensuremath{\mathsf{F_{\sigma}^{R}}}\xspace}
\newcommand{\Gd}{\ensuremath{\mathsf{G}_{\delta}}\xspace}
\newcommand{\GdR}{\ensuremath{\mathsf{G_{\delta}^{R}}}\xspace}
\newcommand{\RAT}{\ensuremath{\mathsf{RAT}}\xspace}

\newdefinition{definition}{Definition}[section]
\newdefinition{remark}{\normalfont \it Remark}
\newdefinition{example}{\normalfont \it Example}
\newtheorem{theorem}{Theorem}[section]
\newtheorem{proposition}[theorem]{Proposition}
\newtheorem{lemma}[theorem]{Lemma}
\newtheorem{corollary}[theorem]{Corollary}

\newproof{proof}{\textit{Proof}}

\makeatletter
\def\ps@pprintTitle{%
  \let\@oddhead\@empty
  \let\@evenhead\@empty
  \def\@oddfoot{\reset@font\hfil\thepage\hfil}
  \let\@evenfoot\@oddfoot
}
\makeatother

\begin{document}

\begin{frontmatter}

\title{Acceptance conditions for \ensuremath{\omega}-languages\\ and the Borel hierarchy\tnoteref{PREV} \tnoteref{ANR}}

\tnotetext[PREV]{A preliminary version of this paper was accepted for presentation at DLT'2012 conference~\cite{dennunzio2012}.}

\tnotetext[ANR]{This work has been partially supported by the French National Research Agency project EMC (ANR-09-BLAN-0164) and by PRIN/MIUR project ``Mathematical aspects and forthcoming applications of automata and formal languages''.}

\author[Paris]{Julien Cervelle}
\ead{julien.cervelle@polytechnique.edu}

\author[Milano]{Alberto Dennunzio\corref{cor}}
\ead{dennunzio@disco.unimib.it}

\author[Nice]{Enrico Formenti\corref{cor}}
\ead{enrico.formenti@unice.fr}

\author[Giessen]{Julien Provillard}
\ead{julien.provillard@i3s.unice.fr}

\cortext[cor]{Corresponding author.}

\address[Paris]{LACL UFR de Sciences et Technologie
Universit\'e Paris-Est Cr\'eteil Val-de-Marne,
61 avenue du G\'en\'eral de Gaulle,
94010 Cr\'eteil cedex,
France}

\address[Milano]{Universit\`a degli studi di Milano-Bicocca,
Dipartimento di Informatica Sistemistica e Comunicazione,  viale
Sarca 336, 20126 Milano (Italy)}
\address[Nice]{Universit\'e Nice-Sophia Antipolis,
Laboratoire I3S, 2000 Route des Colles, 06903 Sophia Antipolis
(France)}

\address[Giessen]{
Justus-Liebig Universit\"at Gie\ss en,
Institut f\"ur Informatik,
Arndtstra\ss e 2,
35392 Gie\ss en
}

\begin{abstract}
This paper investigates acceptance conditions for finite automata recognizing $\omega$-regular languages.
As a first result, we show that, under any acceptance condition that can be defined in the MSO logic, a finite automaton can recognize at
most $\omega$-regular languages. Starting from this, the paper aims at
classifying acceptance conditions according to their expressive power and at finding the exact position of the classes of $\omega$-languages they induced according to the Borel hierarchy.
A new interesting acceptance condition is introduced and fully characterized. A step forward is also made in
the understanding of the expressive power of  $(\fin,=)$.
\end{abstract}

\begin{keyword}
finite automata \sep acceptance conditions \sep $\omega$-regular languages \sep Borel hierarchy
\end{keyword}

\end{frontmatter}

\section{Introduction}

Infinite words arose as a natural extension of finite words. Their first usage (at least to our knowledge) was in symbolic dynamics.
Nowadays, they are perused in several scientific domains for example in formal specification and verification of non-terminating processes (e.g. web-servers, OS daemons, \etc) \cite{kurshan1994,kupferman2004,vardi2007}, game theory~\cite{apt2011,thomas2011}, and so on.
\smallskip

In formal software verification, for instance,  the overall state of the system
is represented by an element of some finite alphabet. Hence runs of the systems can be conveniently represented as
$\omega$-words. Finite automata are often used to model the transitions of the system and their accepted
language represents the set of admissible runs of the system under observation. Acceptance conditions
on finite automata are therefore selectors of admissible runs. Main results and overall exposition about  
$\omega$-languages can be found in \cite{thomas1990,staiger1997,perrin2004}.

Seminal studies about acceptance of infinite words by finite automata (\FA) have been
carried out by Richard B\"uchi while investigating monadic second order theories \cite{Buchi1960}.
A B\"uchi automaton \A  accepts an infinite word $w$ if and only if there exists a run of \A which
passes infinitely often through a set of accepting states while reading $w$. Later on, David Muller 
characterized runs that pass through all elements of a given set of accepting states and visit them infinitely 
often \cite{Muller1963}. Afterwards, more acceptance conditions appeared in a series of papers 
\cite{Hartmanis1967,landweber1969,Staiger1974,Moriya1988,Litovsky1997}. Each of these works was trying to
capture a particular semantic on the runs or to fill some conceptual gap.

Acceptance conditions are selectors for runs of the automaton under consideration. Of course, the set of selected
runs is also deeply influenced by the structural properties of the \FA : deterministic vs. non-deterministic, complete 
vs. non complete (see for instance \cite{Litovsky1997}).

The main purpose of this paper is to classify the expressive power of acceptance conditions in relation also with
the structural properties of the automaton. The first result bounds the research to the realm of $\omega$-rational
languages: the language recognized by any \FA under any acceptance condition and \wrt to any structural property
are $\omega$-rational.

Afterwards, the paper aims at positioning the classes of languages induced by the acceptance conditions found
in literature using the Borel hierarchy as a backbone. Figure~\ref{fig:hierarchy-before} illustrates the 
current state of art whilst Figure \ref{fig:hierarchy-after} summarizes the results provided by the 
present paper. Figure~\ref{fig:hierarchy-after} also illustrates the position of a new natural acceptance condition,
called \emph{\ninf}, 
introduced in the present paper to complete the panorama. This new acceptance condition declares a run of a \FA
successful if it goes through a set of accepting states only a finitely number of times or never. The underlying semantic
is that of a non-terminating process which has to definitively enter a safe state after a finite number (possibly zero) of 
exceptions (unsafe states). If some of the classes induced by \emph{\ninf} coincide with already known classes of the
Borel hierarchy, others (those induced by $(\ninf,\sqcap)$) constitute a diamond strictly below $\FsR$.  

\section{Notations, background and basic definitions}
For any set $A$, $\card{A}$ denotes the cardinality of $A$. 
Given a finite alphabet $\Sigma$, $\aStar$ and $\aN$ respectively denote the set of all finite words and the set of all infinite 
words on $\Sigma$, respectively. As usual,  $\epsilon \in \aStar$ is the empty word. For any pair $u,v \in \aStar$, 
$uv$ is the concatenation of $u$ with $v$.

A \emph{language} is any set $\cL\subseteq\aStar$. For languages $\cL_1,\cL_2$,  denote $\cL_1\cL_2 = \set{uv \in \aStar : u \in \cL_1, v \in \cL_2}$ the concatenation of $\cL_1$ and $\cL_2$. For a language $\cL\subseteq\aStar$, denote $\cL^{0}=\set{\epsilon}$, $\cL^{n+1} = \cL^{n}\cL$ and $\cL^{*} = \bigcup_{n \in \N} \cL^n$ the Kleene star of $\cL$. 
The class of \emph{rational languages} 
is the smallest class of languages containing $\emptyset$, all sets $\set{a}$ (for $a \in \Sigma$) and which is closed by union, concatenation and Kleene star.

An \emph{\Nlanguage} is any subset of $\aN$. For a language $\cL$, the infinite iteration of $\cL$ is the \Nlanguage
$$\cL^{\omega} = \set{x \in \aN : \exists (u_i)_{i \in \N} \in (\cL \smallsetminus \set{\epsilon})^{\N}, x = u_0u_1u_2\dots} \enspace.$$ A \Nlanguage $\cL$ is \emph{\Nrational} if there exist two families $\{\cL_i\}$ and $\{\cL^\prime_i\}$ of rational languages such that $\cL = \bigcup_{i=0}^n \cL^\prime_i {\cL_i}^{\omega}$. Denote by \RAT the set of all \Nrational languages.

A \emph{finite automaton} ($\FA$) is a tuple $(\Sigma,Q,T,q_0,\calF)$ where $\Sigma$ is a finite alphabet, $Q$ a finite set of states, $T \subseteq Q \times \Sigma \times Q$ is the set of \emph{transitions}, $q_0 \in Q$ is the \emph{initial state} and $\calF \subseteq \powerset{Q}$ is the \emph{acceptance table}. A $\FA$ is a \emph{deterministic} finite state automaton ($\DFA$) if $\card{\set{q \in Q : (p,a,q) \in T}} \leq 1$ for all $p \in Q$, $a \in \Sigma$. It is a \emph{complete} finite state automaton ($\CFA$) if $\card{\set{q \in Q : (p,a,q) \in T}} \geq 1$ for all $p \in Q$, $a \in \Sigma$. We write $\CDFA$ for a $\FA$ which is both deterministic and complete. 
An (infinite) \emph{path} in a FA $\aut=(\Sigma,Q,T,q_0,\calF)$ is a sequence $(p_i,x_i,p_{i+1})_{i \in\N}$ such that $(p_i,x_i,p_{i+1})\in T$ for all $i\in \N$. The (infinite) word $(x_i)_{i \in\N}$ is the \emph{label} of the path $p$. A path is said to be \emph{initial} if $p_0=q_0$.
\begin{definition}
Let $\aut = (\Sigma,Q,T,q_0,\calF)$ be a $\FA$ and $p = (p_i, x_i, q_i)_{i \in \N}$ an infinite path in $\aut$. Define the sets
\begin{itemize}
\item
$\run_{\aut}(p) = \{q \in Q : \exists i > 0, p_{i} = q \}$,
\item
$\inf_{\aut}(p) = \{q \in Q : \forall i > 0, \exists j \geq i, p_{j} = q \}$,
\item
$\fin_{\aut}(p) = \run(p) \smallsetminus \inf(p)$,
\item
$\ninf_{\aut}(p) = Q \smallsetminus \inf(p)$
\end{itemize}
as the sets of states \emph{appearing at least one time, infinitely many times,  finitely many times but at least once, and either finitely many times or never} in $p$, respectively. 
\end{definition}

An \emph{acceptance condition} is a subset of all the initial infinite paths. The paths inside such a subset are called \emph{accepting paths}. Let $\aut$ be a $\FA$ and $cond$ be an acceptance condition for $\aut$, a word $w$ is \emph{accepted} by $\aut$ (under condition $cond$) if and only if it is the label of some accepting path.

Let $\sqcap$ be the binary relation over sets such that for all sets $A$ and $B$, $A \sqcap B$ if and only if $A \cap B \neq \emptyset$. 

In the sequel, we will consider acceptance conditions induced by pairs $(c,\textbf{R}) \in \set{\run, \inf, \fin, \ninf} \times \set{\sqcap, \subseteq, =}$. A pair $cond=(c,\textbf{R})$ defines an acceptance condition $cond_{\aut}$ on an automaton $\aut = (\Sigma,Q,T,i,\calF)$ as follows: an initial path $p = (p_i, a_i, p_{i+1})_{i \in \N}$ is  accepting if and only if there exists a set $F \in \calF$ such that  $c_{\aut}(p)~\textbf{R}~F$. We denote by $\Lcond{cond}{\aut}$ the \emph{language accepted by $\aut$ under the acceptance condition $cond_{\aut}$}, \ie, the set of all words accepted by $\aut$ under $cond_{\aut}$.

\begin{definition}
\label{classes_of_language}
For any pair $cond=(c,\textbf{R}) \in \set{\run, \inf, \fin, \ninf} \times \set{\sqcap, \subseteq, =}$ and for any finite alphabet $\Sigma$, define the following sets 
\begin{itemize}
\item
$\FA^{(\Sigma)}(cond) = \set{\Lcond{cond}{\aut}, \text{ \aut is a $\FA$ on $\Sigma$}}$,
\item
$\DFA^{(\Sigma)}(cond) = \set{\Lcond{cond}{\aut},\text{ \aut is a $\DFA$ on $\Sigma$}}$,
\item
$\CFA^{(\Sigma)}(cond) = \set{\Lcond{cond}{\aut},\text{ \aut is a $\CFA$ on $\Sigma$}}$,
\item
$\CDFA^{(\Sigma)}(cond) = \set{\Lcond{cond}{\aut}, \text{ \aut is a $\CDFA$ on $\Sigma$}}$
\end{itemize}
as the classes of languages accepted by $\FA$, $\DFA$, $\CFA$, and $\CDFA$, respectively, under the acceptance condition derived by $cond$.
\end{definition}

Some of the acceptance conditions derived by pairs $(c,\textbf{R})$ have been studied in the literature as summarized in the Table~\ref{known_results}.
\begin{table}[htb]
\begin{center}\small
\begin{tabular}{|c|c|c|c|}
\hline
& $\sqcap$ & $ \subseteq $ & $=$ \\
\hline
$\run$ & Landweber \cite{landweber1969} & Hartmanis\;\&\;Stearns \cite{Hartmanis1967} & Staiger\;\&\;Wagner \cite{Staiger1974}\\
\hline
$\inf$ & B\"uchi \cite{Buchi1960} & Landweber \cite{landweber1969} & Muller \cite{Muller1963}\\
\hline
$\fin$ & Litovski\;\&\;Staiger \cite{Litovsky1997} &  \textsc{this paper} (partially) & \textsc{this paper}\footnotemark[2]\\
\hline
$\ninf$ & \textsc{this paper}\footnotemark[1] & \textsc{this paper}\footnotemark[1] & \textsc{this paper}\\
\hline
\end{tabular}
\end{center}
\caption{Known results on acceptance conditions.}
\label{known_results}
\end{table}
\footnotetext[1]{These conditions have been already investigated in \cite{Moriya1988} but only in the case of complete automata with a unique set of accepting states.}
\footnotetext[2]{Only $\FA$ and $\CFA$ are considered here. For $\DFA$ and $\CDFA$ the question is still open.}

For $\Sigma$ endowed with discrete topology and
$\aN$ with the induced product topology, let $\F$, $\G$, $\Fs$ and $\Gd$ be the collections of all closed sets, open sets, countable unions of closed set and countable intersections of open sets, respectively. For any pair $A,B$ of collections of sets, denote by $\bool{A}$, $A~\Delta~B$, and $A^\mathsf{R}$ the  boolean closure of $A$, the set $\set{U \cap V : U \in A, V \in B}$ and the set $A \cap \RAT$, respectively.
These, indeed, are the lower classes of the Borel hierarchy. For more on this subject we refer the reader
to \cite{Wagner1979} or \cite{perrin2004}, for instance.

\begin{remark}
Rational and $\FsR$ sets are stable by projection.
\end{remark}

From now on, we fix a finite alphabet $\Sigma$ and we omit to mention it in classes of languages.
Figure~\ref{fig:hierarchy-before} illustrates the known hierarchy of languages classes (arrows represents strict inclusions).

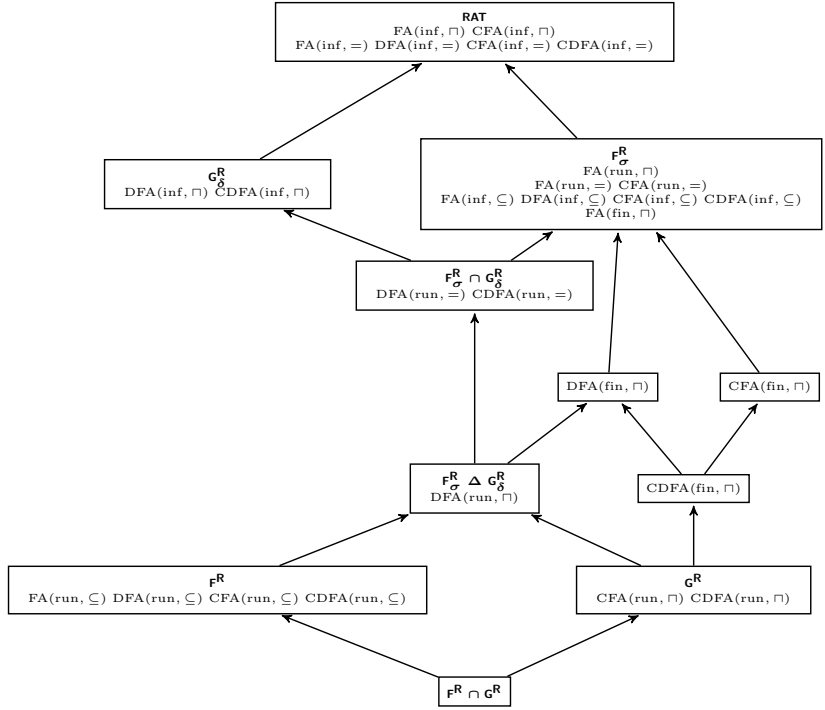
\begin{figure}[htb]
\begin{center}
\scalebox{0.9}{
\begin{tikzpicture}[semithick, shorten >=1pt, >=stealth']
\newcommand{\esph}{4.5}
\newcommand{\espv}{1.5}
\newcommand{\myxshift}{-20}

\tikzstyle{vertex}=[draw, shape=rectangle, font=\tiny]

\node[vertex,xshift=\myxshift] (A) at (0*\esph,0.5*\espv)
{
    $\begin{array}{c}
    \boldsymbol{\RAT} \\
    \FA(\inf,\sqcap)~\CFA(\inf,\sqcap)\\
    \FA(\inf,=)~\DFA(\inf,=)~\CFA(\inf,=)~\CDFA(\inf,=)
    \end{array}$
};
\node[vertex,xshift=\myxshift-28] (B) at (.7*\esph,-1*\espv)
{
    $\begin{array}{c}
    \boldsymbol{\FsR} \\
    \FA(\run,\sqcap)\\
    \FA(\run,=)~\CFA(\run,=)\\
    \FA(\inf,\subseteq)~\DFA(\inf,\subseteq)~\CFA(\inf,\subseteq)~\CDFA(\inf,\subseteq) \\
    \FA(\fin,\sqcap)
    \end{array}$
};
\node[vertex] (C) at (-1*\esph,-\espv) 
{
    $\begin{array}{c}
    \boldsymbol{\GdR} \\
    \DFA(\inf,\sqcap)~\CDFA(\inf,\sqcap)
    \end{array}$
};
\node[vertex,xshift=\myxshift] (D) at (0*\esph,-2*\espv)
{
    $\begin{array}{c}
    \boldsymbol{\FsR \cap \GdR} \\
    \DFA(\run,=)~\CDFA(\run,=)
    \end{array}$
};
\node[vertex] (E) at (-1*\esph,-5*\espv)
{
    $\begin{array}{c}
    \boldsymbol{\FR}\\
    \FA(\run,\subseteq)~\DFA(\run, \subseteq)~\CFA(\run,\subseteq)~\CDFA(\run, \subseteq)
    \end{array}$
};
\node[vertex,xshift=\myxshift-36] (F) at (1*\esph,-5*\espv)
{
    $\begin{array}{c}
    \boldsymbol{\GR}\\
    \CFA(\run,\sqcap)~\CDFA(\run,\sqcap)
    \end{array}$
};
\node[vertex,xshift=\myxshift] (G) at (0*\esph,-6*\espv)
{
    $\boldsymbol{\FR \cap \GR}$
};
\node[vertex,xshift=\myxshift] (H) at (0*\esph,-4*\espv)
{
    $\begin{array}{c}
    \boldsymbol{\FsR~\Delta~\GdR} \\
    \DFA(\run,\sqcap)
    \end{array}$
};
\node[vertex,xshift=\myxshift-36] (I) at (1*\esph,-4*\espv)
{
    $\CDFA(\fin,\sqcap)$
};
\node[vertex,xshift=\myxshift-8] (J) at (0.5*\esph,-3*\espv)
{
    $\DFA(\fin,\sqcap)$
};
\node[vertex,xshift=\myxshift-68] (K) at (1.5*\esph,-3*\espv)
{
    $\CFA(\fin,\sqcap)$
};

\draw[->] (B) -- (A);
\draw[->] (C) -- (A);
\draw[->] (D) -- (B);
\draw[->] (D) -- (C);
\draw[->] (E) -- (H);
\draw[->] (F) -- (H);
\draw[->] (H) -- (D);
\draw[->] (G) -- (E);
\draw[->] (G) -- (F);
\draw[->] (F) -- (I);
\draw[->] (H) -- (J);
\draw[->] (I) -- (J);
\draw[->] (I) -- (K);
\draw[->] (J) -- (B);
\draw[->] (K) -- (B);

\end{tikzpicture}
}
\end{center}
\caption{Currently known relations between classes of $\omega$-languages recognized by \FA
according to the considered acceptance conditions and structural properties like determinism or completeness. Classes of the Borel hierarchy are typeset in bold. Arrows mean strict inclusion. Classes in the same box coincide.}
\label{fig:hierarchy-before}
\end{figure}

\section{A turn into logic}
In~\cite{Buchi1960}, B\"uchi showed that a \Nlanguage is rational if and only if it is definable in the MSO logic. We show that all the  languages recognized by one of the previously introduced acceptance condition are MSO-definable and hence rational. More generally, if an acceptance condition can be defined in the MSO logic, the languages it allows to recognize are rational.

The \emph{monadic second-order logic} (MSO logic) on the alphabet $\Sigma$ is the logical system defined by
\begin{itemize}
\item
first-order variables $x$, $y$, $z$ \dots
\item
second-order variables (of arity 1) $X$, $Y$, $Z$ \dots
\item
unary relations $Q_a$ for $a \in \Sigma$,
\item
and the binary relations $=$, $S$ et $<$.
\end{itemize}

The \emph{atomic formulas} are formulas of the form
\[
x = y,\hspace{0.2cm} X(x),\hspace{0.2cm} S(x,y),\hspace{0.2cm} x < y,\hspace{0.2cm} Q_a(x)
\]
where $x$ and $y$ are first-order variables, $X$ is a second-order variable and $a \in \Sigma$.

The set of \emph{second-order formulas} is the smallest set which contains atomic formulas and such that for all second-order formulas $\phi$ and $\psi$, for all first-order variables $x$, for all second-order variables $X$,
\[
\neg \psi, \hspace{0.2cm}\phi \vee \psi,\hspace{0.2cm} \phi \wedge \psi,\hspace{0.2cm} \phi \rightarrow \psi, \hspace{0.2cm}
\exists x\phi,\hspace{0.2cm} \forall x\phi,\hspace{0.2cm} \exists X \phi,\hspace{0.2cm} \forall X \phi
\]
are second-order formulas.

A variable is \emph{free} in a formula if it is not introduced by a quantifier. If $\phi$ is a formula, we denote by $\FV{\phi}$ the set of free variables which occur in $\phi$. This set is recursively defined by
\begin{itemize}
\item
$\FV{x = y} = \FV{S(x,y)} = \FV{x < y} = \set{x,y}$,
\item
$\FV{X(x)} = \set{X,x}$,
\item
$\FV{Q_a(x)} = \set{x}$,
\item
$\FV{\neg \phi} = \FV{\phi}$,
\item
$\FV{\phi \vee \psi} = \FV{\phi \wedge \psi} = \FV{\phi \rightarrow \psi} = \FV{\phi} \cup \FV{\psi}$,
\item
$\FV{\exists x \phi} = \FV{\forall x \phi} = \FV{\phi} \smallsetminus \set{x}$ and
\item
$\FV{\exists X \phi} = \FV{\forall X \phi} = \FV{\phi} \smallsetminus \set{X}$
\end{itemize}
for all first-order variables $x$ and $y$, for all second-order variable $X$ and for all formulas $\phi$ and $\psi$.

A \emph{closed formula} is a formula without free variables. We usually denote by $\phi(X_1,\dots,X_m,x_1,\dots,x_n)$ a formula $\phi$ where at most the variables $X_1,\dots,X_m$ and $x_1,\dots,x_n$ occur free.

\begin{definition}
Let $w$ be an infinite word on $\Sigma$, $E_1, \dots, E_m \subseteq \N$, $i_1, \dots, i_n \in \N$ and  $\phi(X_1,\dots,X_m,x_1,\dots,x_n)$ a formula. The word $w$ \emph{satisfies} the formula $\phi$, which is denoted by
\[
(w, E_1, \dots, E_m, i_1, \dots, i_n) \models \phi(X_1, \dots ,X_m, x_1, \dots, x_n) \enspace,
\]
if $\phi$ is true when
\begin{itemize}
\item
first-orders variables are interpreted as naturals,
\item
second-orders variables are interpreted as subsets of $\N$,
\item
$\forall a \in \Sigma$, $Q_a$ is interpreted as the set $\set{i \in \N : w_i = a}$,
\item
the unary relations are interpreted as the membership relations to the corresponding sets,
\item
the relations $=$, $S$ et $<$ are interpreted to be the equality, successor and order relations on $\N$, respectively,
\item
$E_j$ is the interpretation of $X_j$ for $j \in [1,m]$,
\item
$i_j$ is the interpretation of $x_j$ for $j \in [1,n]$.
\end{itemize}
\end{definition}

\begin{definition}
Let $\phi$ be a statement, the \emph{language} of $\phi$ is the set
\[
\calL_\phi = \set{w \in \Sigma^\omega : w \models \phi}
\]
of all $\omega$-words satisfying $\phi$.

A \Nlanguage $\calL \subseteq \Sigma^\omega$ is \emph{MSO-definable} if there exists a closed formula $\phi$ such that $\calL = \calL_\phi$.
\end{definition}

\begin{theorem}[B\"uchi \cite{Buchi1960}]
\label{th:buchi}
A \Nlanguage is \Nrational if and only if it MSO-definable.
\end{theorem}

\begin{proposition}
\label{prop:all_rat}
Let $\aut=(\Sigma,Q,T,q_0,\calF)$ be a \FA and $cond$ an acceptance condition derived by a pair $(c,\textbf{R}) \in \set{\run, \inf, \fin, \ninf} \times \set{\sqcap, \subseteq, =}$, then $\Lcond{\aut}{cond}$ is \Nrational.
\end{proposition}

\begin{proof}
We prove that the language $\Lcond{\aut}{cond}$ is MSO-definable and we conclude by using Theorem~\ref{th:buchi}. We construct a formula $\phi$ which encodes the automaton on one hand and the acceptance condition on the other hand. Let $n = \card{Q}$ and let $q_0, \dots, q_{n-1}$ denote the elements in $Q$.  The formula describing the language is\footnote{By convention $\displaystyle{\bigvee_{i \in \emptyset} \phi_i = \mathrm{\textbf{false}}}$ and $\displaystyle{\bigwedge_{i \in \emptyset} \phi_i = \mathrm{\textbf{true}}}$.}
\begin{align*}
\phi =  \exists X_{q_0} & \dots \exists X_{q_{n-1}}\\
& \bigg( \bigwedge_{p,q \in Q, p \neq q} \neg \exists x\; \Big( X_p(x) \wedge X_q(x)  \Big) \bigg) \wedge \\
& \bigg( \forall x \forall y\; S(x,y) \rightarrow \bigvee_{(p,a,q) \in T} \Big( X_p(x) \wedge Q_a(x) \wedge X_q(y) \Big) \bigg) \wedge \\
& \bigg( \exists x\; \Big( \neg \exists y\; S(y,x) \Big) \wedge X_{q_0}(x) \bigg) \wedge \mathrm{COND}(X_{q_0}, \dots, X_{q_{n-1}}) \enspace.
\end{align*}

The first three lines encode a path in $\cA$.
For such a path $(p_i,a_i,p_{i+1})_{i \in \N}$,  the variable~$X_{q}$ will represent the set $\set{i \in \N : p_i = q}$.
The formula
\[
\bigwedge_{p,q \in Q, p \neq q} \neg \exists x\; \Big( X_p(x) \wedge X_q(x)  \Big)
\]
enforces the sets $X_q$ to be pairwise disjoint, whereas the formula
\[
\forall x \forall y\; S(x,y) \rightarrow \bigvee_{(p,a,q) \in T} \Big( X_p(x) \wedge Q_a(x) \wedge X_q(y) \Big)
\]
indicates that a transition $(p,a,q) \in T$ has to be used to go from a state $p$ to a state $q$ by reading a letter $a$. The formula $\exists x\; \Big( \neg \exists y\; S(y,x) \Big) \wedge X_{q_0}(x)$ enforces the path to be initial because 0 is the only integer which does not have a predecessor and it has to start in the state $q_0$ in this case. Finally, the formula $\mathrm{COND}(X_{q_0}, \dots, X_{q_{n-1}})$ encodes the fact that the path is accepting according to the considered acceptance condition and its expression depends on the pair $(c,\textbf{R})$ as we will see in the following. Let $C(X)$ be the formula defined by
\[
C(X) :=
\begin{cases}
\exists x\; \Big( \exists y\; S(y,x) \Big) \wedge X(x) & \text{if $c =\run$} \\
\forall x \exists y\; (x < y) \wedge X(y) & \text{if $c = \inf$} \\
\begin{aligned}
\Big( \exists x\; \Big( \exists y\; S(y,x) \Big) \wedge X(x) \Big) \wedge \phantom{vvvv} \\ \Big( \neg \forall x \exists y\; (x < y) \wedge X(y) \Big)
\end{aligned}
 & \text{if $c = \fin$} \\
\neg \forall x \exists y\; (x < y) \wedge X(y) & \text{if $c = \ninf$} \\
\end{cases}
\enspace.
\]
For all $q \in Q$, the formula $C(X_q)$ would be true if and only if the previously encoded path $p$ verifies $q \in c_\aut(p)$.

We can now write the formula $\mathrm{COND}(X_{q_0}, \dots, X_{q_{n-1}})$ depending on $\textbf{R}$ by
\begin{itemize}
\item
for the relation $\sqcap$,
\[
\bigvee_{F \in \calF} \bigvee_{q \in F} C(X_q) \enspace,
\]
\item
for the relation $\subseteq$,
\[
\bigvee_{F \in \calF} \bigwedge_{q \in Q\smallsetminus F} \neg C(X_q) \enspace,
\]
\item
for the relation $=$,
\[
\bigvee_{F \in \calF} \left( \bigwedge_{q \in F} C(X_q) \wedge \bigwedge_{q \in Q\smallsetminus F} \neg C(X_q) \right) \enspace.
\] 
\end{itemize}
\qed
\end{proof}

Using the same proof, we can show that any acceptance condition which is MSO-definable only induces rational languages. We have just to change the formula $\mathrm{COND}(X_{q_0}, \dots, X_{q_{n-1}})$ to fit to the acceptance condition.

\section{The acceptance conditions $\mathbb{A}$ and \textbf{$\mathbb{A}'$} and the Borel hierarchy}

In \cite{Moriya1988},  Moriya and Yamasaki introduced two more acceptance conditions, namely $\mathbb{A}$ and
$\mathbb{A}'$, and they compared them to the Borel hierarchy for the case of \CFA and \CDFA having a
unique set of accepting states. In this section, those results are generalized to \FA and \DFA and to any
set of sets of accepting states.

\begin{definition}
Given a $\FA$ $\aut=(\Sigma,Q,T,q_0,\calF)$, the
acceptance condition $\mathbb{A}$ (resp. $\mathbb{A'}$) on $\aut$  is defined as follows: an initial path $p$ is accepting under $\mathbb{A}$ (resp. $\mathbb{A}'$) if and only if there exists a set 
$F \in \calF$ such that $F \subseteq \run_{\aut}(p)$ (resp. $F \not\subseteq \run_{\aut}(p)$).
\end{definition}

We denote by $\Lcond{\mathbb{A}}{\aut}$ (resp. $\Lcond{\mathbb{A}'}{\aut}$) the language accepted by an automaton $\aut$ under the acceptance condition $\mathbb{A}$ (resp. $\mathbb{A}$'). Similar notation as Definition~\ref{classes_of_language} are used for classes of languages.

\begin{lemma}\label{aap}\mbox{\phantom{Q}}\\[-4ex]
\begin{enumerate}
\item $\FA(\mathbb{A})\subseteq \FA(\run, \sqcap)$\enspace, 
\item $\DFA(\mathbb{A})\subseteq \DFA(\run, \sqcap)$\enspace, 
\item $\CFA(\mathbb{A})\subseteq \CFA(\run, \sqcap)$\enspace,
\item $\CDFA(\mathbb{A})\subseteq \CDFA(\run, \sqcap)$\enspace.
\end{enumerate}
\end{lemma}

\begin{proof}
We are going to prove that for any $\FA$ $\aut=(\Sigma,Q,T,q_0,\calF)$, there exists an automaton $\aut'$ such that $\Lcond{\mathbb{A}}{\aut} = \Lcond{(\run,\sqcap)}{\aut'}$ and $\aut'$ is deterministic (resp. complete) if $\aut$ is deterministic (resp. complete).

Let $\aut' = (\Sigma,Q \times \powerset{Q},T',(q_0, \emptyset),\calF')$ where
\[
T' = \set{((p,S),a,(q, S \cup \set{q})) : (p, a , q) \in T, S \in \powerset{Q}}
\]
and
\[
\calF' = \set{\set{(q, S)} : q \in Q, S \in \powerset{Q}, \exists F \in \calF, F \subseteq S} \enspace.
\]
Clearly, $\aut'$ is deterministic (resp. complete) if \A is deterministic (resp. complete).

We now show that $\Lcond{\mathbb{A}}{\aut}\subseteq \Lcond{(\run,\sqcap)}{\aut'}$. Let $x \in \Lcond{\mathbb{A}}{\aut}$. Then, there exist an initial path $p = (p_i, x_i, p_{i+1})_{i \in \N}$ in $\aut$ and a set  $F \in \calF$ such that $F \subseteq \run_{\aut}(p)$. So, the sequence 
\[p' = ((p_i, \bigcup_{0 < j \leq i } \{p_j\}), x_i, (p_{i+1},\bigcup_{0 < j \leq i+1} \{p_j\}))_{i \in \N}
\]
is an initial path in $\aut'$ with label $x$. Since $Q$ is finite,  $\run_{\aut}(p)=\bigcup_{0 < j \leq n} \{p_j\}$ for some $n\in\N$, it holds that $\set{(p_n, \run_{\aut}(p))}\in \run_{\aut'}(p')$. 
Let $F'=\set{(p_n, \run_{\aut}(p))} \in \calF'$, then $\run_{\aut'}(p')\cap F'\neq\emptyset$ and $x \in \Lcond{(\run,\sqcap)}{\aut'}$.

To prove $\Lcond{(\run,\sqcap)}{\aut'}\subseteq\Lcond{\mathbb{A}}{\aut}$,  let $x \in\Lcond{(\run,\sqcap)}{\aut'}$. Then, there exists an initial path $p' = ((p_i,S_i), x_i, (p_{i+1},S_{i+1}))_{i \in \N}$ in $\aut'$ and a set $F'=\{(q,S)\}\in\calF'$ such that $\run_{\aut'}(p')\cap F'\neq\emptyset$, and, so, there exists a set $F\in\calF$ with $F \subseteq S$ and $S=\bigcup_{0<j\leq k} \{p_j\}$ for some $k\in\N$. Therefore, $p = (p_i,x_i,q_i)_{i \in \N}$ is an initial path with label $x$ in $\aut$ such that $F \subseteq S\subseteq \run(p)$. Hence, $x \in \Lcond{\mathbb{A}}{\aut}$.\qed
\end{proof}

\begin{lemma}\label{apa}\mbox{\phantom{Q}}\\[-4ex]
\begin{enumerate}
\item
$\FA(\run, \sqcap)\subseteq \FA(\mathbb{A})$\enspace, 
\item
$\DFA(\run, \sqcap)\subseteq \DFA(\mathbb{A})$\enspace, 
\item
$\CFA(\run, \sqcap)\subseteq \CFA(\mathbb{A})$\enspace,
\item
$\CDFA(\run, \sqcap)\subseteq \CDFA(\mathbb{A})$\enspace.
\end{enumerate}
\end{lemma}

\begin{proof}
We are going to show that for any $\FA$ $\aut=(\Sigma,Q,T,q_0,\calF)$ there exists an $\FA$ $\aut'$ such that $ \Lcond{(\run,\sqcap)}{\aut}=\Lcond{\mathbb{A}}{\aut'} $ and $\aut'$ is deterministic (resp. complete) if $\aut$ is deterministic (resp. complete). 

Let $\aut' = (\Sigma,Q,T,q_0,\calF')$ where $\calF' = \set{\set{q} : q \in Q, \exists F \in \calF, q \in F}$. Clearly, $\aut'$ is deterministic (resp. complete) if $\aut$ is deterministic (resp. complete). Moreover, $x \in \Lcond{(\run,\sqcap)}{\aut}$ if and only if there exist an initial path $p$ in $\aut$ with label $x$ and a set $F \in \calF$ such that $\run_{\aut}(p)\cap F\neq\emptyset$, or, equivalently,  there exist an initial path $p$ in $\aut'$ with label $x$ and a set $F' \in \calF'$ such that $F' \subseteq \run_{\aut'}(p)$, \ie, if and only if $x \in \Lcond{\mathbb{A}}{\aut'}$.\qed
\end{proof}

\begin{lemma}\label{aprime}\mbox{\phantom{Q}}\\[-4ex]
\begin{enumerate}
\item
$\FA(\mathbb{A'})\subseteq \FA(\run,\subseteq)$\enspace, 
\item
$\DFA(\mathbb{A'})\subseteq \DFA(\run, \subseteq)$\enspace,
\item
$\CFA(\mathbb{A'})\subseteq \CFA(\run, \subseteq)$\enspace, 
\item
$\CDFA(\mathbb{A'})\subseteq \CDFA(\run, \subseteq)$\enspace.
\end{enumerate}
\end{lemma}

\begin{proof}
We are going to show that for any $\FA$ $\aut=(\Sigma,Q,T,q_0,\calF)$  there exists an automaton $\aut'$ such that $\Lcond{\mathbb{A'}}{\aut} = \Lcond{(\run, \subseteq)}{\aut'}$ and $\aut'$ is deterministic (resp. complete) if $\aut$ is deterministic (resp. complete).

Let $\aut' = (\Sigma, Q', T', (q_0,\emptyset), \calF')$ where $Q'=(Q \times \powerset{Q}) \cup \set{\bot}$, $\calF'=\powerset{Q \times \powerset{Q}}$, and
\begin{align*}
T' & =  \left\{((p,S),a,(q, S \cup \set{q})) : (p, a , q) \in T, S \in \powerset{Q}, \exists F \in \calF, F \not\subseteq S \cup \set{q}\right\} \\
 & \bigcup  \set{((p,S),a,\bot) : S \in \powerset{Q}, \exists q \in Q, (p, a , q) \in T, \forall F \in \calF, F \subseteq S \cup \set{q}} \\
 & \bigcup  \set{(\bot, a, \bot) : a \in \Sigma} \enspace.
\end{align*}
Then, 
$\A'$ is deterministic (resp. complete) if \A is deterministic (resp. complete). The state $\bot$ acts as a sink for $\A'$ and it is reached as soon as it is no more possible to not contain a set in the acceptance table for the corresponding path in \A.
Indeed, 
$x \in \Lcond{\mathbb{A}'}{\aut}$ if and only if there exist an initial path $p$ in $\aut$ with label $x$ and a set $F \in \calF$ such that $F \not\subseteq \run_{\aut}(p)$ if and only if there exists an initial path $p'$ in $\aut'$ with label $x$ such that $p'_n \neq \bot$ for all $n\in\N$, \ie, if and only if $x \in \Lcond{(\run, \subseteq)}{\aut'}$.
\qed
\end{proof}

\begin{lemma}\label{prop_run_inclusion}\mbox{\phantom{Q}}\\[-4ex]
\begin{enumerate}
\item
$\FA(\run, \subseteq)\subseteq \FA(\mathbb{A'})$\enspace,
\item
$\DFA(\run, \subseteq)\subseteq \DFA(\mathbb{A'})$\enspace,
\item
$\CFA(\run, \subseteq)\subseteq \CFA(\mathbb{A'})$\enspace,
\item
$\CDFA(\run, \subseteq)\subseteq \CDFA(\mathbb{A'})$\enspace.
\end{enumerate}
\end{lemma}

\begin{proof}
We are going to show that for any $\FA$ $\aut=(\Sigma,Q,T,q_0,\calF)$ there exists an automaton $\aut'$  such that $\Lcond{\mathbb{A'}}{\aut'} = \Lcond{(\run, \subseteq)}{\aut}$ and $\aut'$ is deterministic (resp. complete) if $\aut$ is deterministic (resp. complete). 

Let $\aut' = (\Sigma, Q', T', (q_0,\emptyset), \calF')$ where $Q'=(Q \times \powerset{Q}) \cup \set{\bot}$, $\calF'=\set{\set{\bot}}$, and 
\begin{align*}
T' & =  \set{((p,S),a,(q, S \cup \set{q})) : (p, a , q) \in T, S \in \powerset{Q}, \exists F \in \calF, S \cup \set{q} \subseteq F} \\
 & \bigcup  \set{((p,S),a,\bot) : S \in \powerset{Q}, \exists q \in Q, (p, a , q) \in T, \forall F \in \calF, S \cup \set{q} \not\subseteq F} \\
 & \bigcup  \set{(\bot, a, \bot) : a \in \Sigma}
\end{align*}

Then, $\aut'$ is deterministic (resp. complete) if $\aut$ is deterministic (resp. complete).
Moreover, 
$x \in \Lcond{(\run, \subseteq)}{\aut}$ if and only if there exists an initial path $p$ in $\aut$ with label $x$ and a set $F \in \calF$ such that $\run_{\aut}(p) \subseteq F$ iff there exists an initial path $p'$ in $\aut'$ with label $x$ such that $p'_n \neq \bot$ for all $n\in\N$, \ie, if and only if $x \in \Lcond{\mathbb{A}'}{\aut'}$.\qed
\end{proof}

The following result places the classes of langages characterized by $\mathbb{A}$ and $\mathbb{A'}$
\wrt the Borel hierarchy.

\begin{theorem}\mbox{\phantom{Q}}\\[-4ex]
\begin{enumerate}
\item
$\CDFA(\mathbb{A}) = \CFA(\mathbb{A}) = \GR$\enspace,
\item
$\DFA(\mathbb{A}) = \FsR~\Delta~\GdR$\enspace,
\item
$\FA(\mathbb{A}) = \FsR$\enspace,
\item
$\CDFA(\mathbb{A}') = \DFA(\mathbb{A}') = \CFA(\mathbb{A}') = \FA(\mathbb{A}') = \FR$\enspace.
\end{enumerate}
\end{theorem}
\begin{proof}
It is a consequence of Lemmata~\ref{aap}, ~\ref{apa}, ~\ref{aprime} and~\ref{prop_run_inclusion},
and the known results (see Figure~\ref{fig:hierarchy-before}) on the classes of languages accepted by $\FA$, $\DFA$, $\CFA$, and $\CDFA$ under the acceptance conditions derived by $(\run, \sqcap)$ and $(\run, \subseteq)$.\qed
\end{proof}

\begin{remark}
Languages in $\CDFA(\mathbb{A})$ (resp. $\CDFA(\mathbb{A'})$) are unions of languages in the class $\mathbb{A}$ (resp. $\mathbb{A'}$) of \cite{Moriya1988}. This class equals $\GR$ (resp. $\FR$) and is closed under union operation.
These facts already prove $\CDFA(\mathbb{A}) = \GR$ 
(resp. $\CDFA(\mathbb{A'}) = \FR$).
\end{remark}

\section{The acceptance conditions $(\ninf,\sqcap)$ and $(\ninf,\subseteq)$.}

In \cite{Litovsky1997}, Litovsky and Staiger studied the class of languages accepted by \FA under the
acceptance condition $(\fin,\sqcap)$ \wrt which a path is successful if it visits an accepting state
finitely many times but at least once. It is natural to study the expressivity
of the similar acceptance condition for which a path is successful if it visits an accepting state finitely many
times or never: $(\ninf,\sqcap)$. The expressivity of $(\ninf,\subseteq)$ is also 
analized and compared with the previous ones to complete the picture in Figure~\ref{fig:hierarchy-before}. 
\smallskip
As a first step, we analyze two more acceptance conditions proposed by Moriya and
Yamasaki \cite{Moriya1988}: $\mathbb{L}$ which represents
the situation of a non-terminating process forced to pass through a finite set of
``safe'' states infinitely often and $\mathbb{L}'$ which is the negation of $\mathbb{L}$.
Lemma \ref{lem:L-ninf-exists} proves that $\mathbb{L}$ is equivalent to $(\ninf,\sqcap)$
and $\mathbb{L}'$ to $(\ninf,\subseteq)$. Moreover, the results of \cite{Moriya1988} are
extended to any type of \FA with any number of sets of accepting states.

\begin{definition}
Given an $\FA$ $\aut=(\Sigma,Q,T,q_0,\calF)$,
the acceptance condition $\mathbb{L}$ (resp. $\mathbb{L'}$) on $\aut$ is defined as follows: an initial path $p$ is accepting under $\mathbb{L}$ (resp. $\mathbb{L}'$) if and only if there exists a set 
$F \in \calF$ such that $F \subseteq \inf_{\aut}(p)$ (resp. $F \not\subseteq \inf_{\aut}(p)$).
\end{definition}

We denote by $\Lcond{\mathbb{L}}{\aut}$ (resp. $\Lcond{\mathbb{L}'}{\aut}$) the language accepted by an automaton $\aut$ under the acceptance condition $\mathbb{L}$ (resp. $\mathbb{L}$'). Similar notation as Definition~\ref{classes_of_language} are used for classes of languages.

\begin{lemma}\label{lem:L-ninf-exists}
$\mathbb{L}$ and $(\ninf, \subseteq)$ (resp. $\mathbb{L}'$ and $(\ninf, \sqcap)$) define the same classes of languages.
\end{lemma}

\begin{proof}
For any  automaton  $\A = (\Sigma,Q,T,q_0,\calF)$ let $\A' = (\Sigma,Q,T,q_0,\calF')$, where $\calF'=\set{Q \smallsetminus F : F \in \calF}$.
Clearly, $\aut'$ is deterministic (resp. complete) iff $\aut$ is deterministic (resp. complete).  Moreover, the following equalities hold
\begin{align*} 
\Lcond{\mathbb{L}}{\A} = \Lcond{(\ninf, \subseteq)}{\A'} \text{ and } \Lcond{(\ninf, \subseteq)}{\A} = \Lcond{\mathbb{L}}{\A'} \\ (\text{resp. } \Lcond{\mathbb{L'}}{\A} = \Lcond{(\ninf, \sqcap)}{\A'} \text{ and } \Lcond{(\ninf, \sqcap)}{\A} = \Lcond{\mathbb{L'}}{\A'})\enspace.
\end{align*}
Hence, the thesis is true.\qed
\end{proof}

Remark that any $\FA$ can be completed with a sink state without changing the language accepted under  $\mathbb{L}$.  Therefore, the following claim is true.
\begin{lemma}\label{lem:FA-L=CFA-L-DFA-L=CDFA-L}
$\FA(\mathbb{L}) = \CFA(\mathbb{L})$ and $\DFA(\mathbb{L}) = \CDFA(\mathbb{L})$.
\end{lemma}

\begin{proposition}
\label{12}
$\CDFA(\inf, \sqcap) \subseteq \CDFA(\mathbb{L})$ and $\CFA(\inf, \sqcap) \subseteq \CFA(\mathbb{L})$.
\end{proposition}

\begin{proof}
For any $\CDFA$ (resp. $\CFA$) $\A = (\Sigma,Q,T,q_0,\calF)$, define the $\CDFA$ (resp. $\CFA$) $\A' = (\Sigma,Q,T,q_0,\calF')$ where
$\calF'=\{\set{q} : \exists F \in \calF, q \in F\}$. 
Then, it follows that $\Lcond{(\inf,\sqcap)}{\A} = \Lcond{\mathbb{L}}{\A'}$ and this concludes the proof.\qed
\end{proof}

\begin{proposition}\label{prop:CDFA-L-subseteq-CDFA-inf-exists}
$\CDFA(\mathbb{L}) \subseteq \CDFA(\inf, \sqcap)$.
\end{proposition}
\begin{proof}
For any $\CDFA$ $\A = (\Sigma,Q,T,q_0,\calF)$ and any $q \in Q$, define the $\CDFA$ $\aut_{q} = (\Sigma,Q,T,q_0,\set{\set{q}})$. By determinism of  \A, it holds that 
$$\Lcond{\mathbb{L}}{\A} = \bigcup_{F \in \calF} \bigcap_{q \in F} \Lcond{(\inf, \sqcap)}{\A_{q}}\enspace.$$
Since $\CDFA(\inf, \sqcap)$ is stable by finite union and finite intersection~\cite{Buchi1960}, there exists a $\CDFA$ $\A'$ such that $\Lcond{\mathbb{L}}{\A}=\Lcond{(\inf, \sqcap)}{\A'}$. Hence,
$\CDFA(\mathbb{L}) \subseteq \CDFA(\inf, \sqcap)$.\qed
\end{proof}

\begin{theorem} 
The following equalities hold.
\begin{enumerate}
\item
$\CDFA(\ninf,\subseteq) = \DFA(\ninf,\subseteq) = \GdR$\enspace,
\item
$\CFA(\ninf,\subseteq) = \FA(\ninf,\subseteq) = \RAT$\enspace.
\end{enumerate}
\end{theorem}
\begin{proof}
The first equality follows from
Lemmata~\ref{lem:L-ninf-exists}  and~\ref{lem:FA-L=CFA-L-DFA-L=CDFA-L}, Propositions~\ref{prop:CDFA-L-subseteq-CDFA-inf-exists} and~\ref{12} and the known fact that $\DFA(\inf,\sqcap)=\CDFA(\inf,\sqcap)=\GdR$, while the second equality 
follows from Lemmata~\ref{lem:L-ninf-exists} and~\ref{lem:FA-L=CFA-L-DFA-L=CDFA-L}, Propositions~\ref{12} and~\ref{prop:all_rat} and the known fact that $\CFA(\inf,\sqcap)=\RAT$.
\qed
\end{proof}

\begin{lemma}
\label{1statof}
For any automaton $\aut=(\Sigma,Q,T,q_0,\calF)$ there exists an automaton $\aut'=(\Sigma',Q',T',q'_0,\calF')$ such that $\calF'=\set{\set{q'}}$ for some $q'\in Q'$, $\Lcond{\mathbb{L}'}{\aut} = \Lcond{\mathbb{L}'}{\aut'}$,  and $\aut'$ is deterministic (resp. complete) if $\aut$ is deterministic (resp. complete).
\end{lemma}
\begin{proof}
If either $\calF = \emptyset$ or $\calF = \set{\emptyset}$ then the automaton $\aut'$ defined by $\Sigma' = \Sigma$, $Q'=\set{\bot}$, $T'=\set{(\bot,a,\bot) : a \in \Sigma}$, $q'_0=q_0$, and $\calF'=\set{\set{\bot}})$ verifies the statement of the Lemma. Otherwise, set $F=\bigcup_{X \in \calF} X$, choose any $f\in F$, and define the automaton $\aut'$ by $\Sigma' = \Sigma$, $Q'=Q \times \powerset{F}$, $q'_0=(q_0,\emptyset)$, $\calF'=\set{\set{(f,F)}}$, and 
\begin{align*}
T' & = \set{((p,S),a,(q, (S \cup \set{q}) \cap F)) : (p, a , q) \in T, (p,S) \neq (f,F)} \\
 & \bigcup  \set{((f,F),a,(q,\emptyset)) : (f, a , q) \in T}\enspace.
\end{align*}
Then, $\aut'$ is deterministic (resp.  complete) if $\aut$ is deterministic (resp.  complete).
Moreover, $\Lcond{\mathbb{L}'}{\aut}\subseteq \Lcond{\mathbb{L}'}{\aut'}$. Indeed,  if $x \in \Lcond{\mathbb{L}'}{\aut}$,  there exist an initial path $p = (p_i, x_i, p_{i+1})_{i \in \N}$ in $\aut$ with label $x$, a set $X \in \calF$, and a state $s \in X$ such that $s \not\in \inf(p)$. Consider the path $p' = ((p_i, S_i), x_i, (p_{i+1},S_{i+1}))_{i \in \N}$ where $S_0 = \emptyset$ and $S_{i+1} = (S_i \cup \set{q_i}) \cap F$ if $(p_i,S_i) \neq (f,F)$, $\emptyset$ otherwise. Then, $p'$ is an initial path in $\aut'$ with label $x$ in which the state $(f,F)$ appears finitely often in $p'$ since $s$ appears finitely often in $p$. Hence, $x \in \Lcond{\mathbb{L}'}{\aut'}$. Finally, the implication $\Lcond{\mathbb{L}'}{\aut'}\subseteq \Lcond{\mathbb{L}'}{\aut}$ is also true. \qed
\end{proof}

The following series of Lemmata  is useful 
to prove strict inclusions between the
the considered language classes.

\begin{lemma}[Moriya and Yamasaki \cite{Moriya1988}]
\label{lem:(a+b)^*a^{omega}}
$\L = (a+b)^*a^{\omega} \in \CDFA(\mathbb{L}')$.
\end{lemma}

\begin{proof}
$\L = \Lcond{\mathbb{L}'}{\A}$ for the $\CDFA$ $\A$ given in Figure~\ref{fig:(a+b)^*a^omega}.\qed
\end{proof}

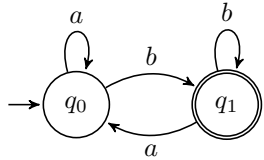
\begin{figure}[htb]
\begin{center}
\scalebox{1.0}{
\begin{tikzpicture}[semithick, shorten >=1pt, node distance=2cm, >=stealth',initial text=]

\node[state, initial] (A) {$q_0$};
\node[state, accepting] (B) [right of=A]{$q_1$};

\path
(A) edge[->, loop above]  node[above] {$a$} (A)
     edge[->, bend left] node[above] {$b$} (B)
(B) edge[->, bend left] node[below] {$a$} (A)
     edge[->, loop above] node[above] {$b$} (B);
\end{tikzpicture}
}
\end{center}
\caption{A $\CDFA$ recognizing $(a+b)^*a^{\omega}$ under $\mathbb{L}'$}
\label{fig:(a+b)^*a^omega}
\end{figure}

\begin{lemma}
\label{lem:ab^*a(a+b)^{omega}}
$ab^*a(a+b)^{\omega} \in \DFA(\mathbb{L}') \smallsetminus \CFA(\mathbb{L}')$.
\end{lemma}

\begin{proof}
Let $\L$ denote the language $ab^*a(a+b)^{\omega}$. Consider the $\DFA$ $\aut'$ in Figure~\ref{fig:ab^*a(a+b)^{omega}}. It is easy to see that $\L = \Lcond{\mathbb{L}'}{\aut'}$.

\begin{figure}
\begin{center}
\scalebox{1.0}{
\begin{tikzpicture}[semithick, shorten >=1pt, node distance=2cm, >=stealth',initial text=]

\node[state, initial] (A) {$q_0$};
\node[state, accepting] (B) [right of=A]{$q_1$};
\node[state] (C) [right of=B]{$q_2$};

\path
(A) edge[->] node[above] {$a$} (B)
(B) edge[->] node[above] {$a$} (C)
    edge[->, loop above] node[above] {$b$} (B)
(C) edge[->, loop above] node[above] {$a, b$} (C);
\end{tikzpicture}
}
\end{center}
\caption{$\DFA$ recognizing $ab^*a(a+b)^{\omega}$ under $\mathbb{L}'$.}
\label{fig:ab^*a(a+b)^{omega}}
\end{figure}
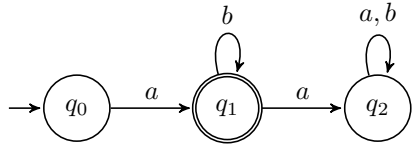

For the sake of argument, suppose that there exists a $\CFA$ $\aut = (\Sigma,Q,T,q_0,\calF)$ such that $\L = \Lcond{\mathbb{L}'}{\aut}$. By Lemma~\ref{1statof}, we can assume that $\calF=\{\{f\}\}$ with $f\in Q$. Let $n = \card{Q}$. Since $ab^{n}a^{\omega} \in \L$ there exists an initial path $p$
and an integer $m$ such that $p_k \neq f$ for all $k > m$.
Since $Q$ is finite, $p_i = p_j$ for some $1 \leq i < j \leq n + 1$ and 
$$(p_0, a, p_1), (p_1, b, p_2),  \ldots (p_i, b, p_{i+1}),\ldots  (p_{j-1}, b, p_j=p_{i})\ldots$$
is an initial path with label $ab^{\omega} \not\in \calL$. Then, $p_h = f$ for some integer $h$ with $i \leq h \leq j$, and, since $\aut$ is complete, there exists an initial path $p' = (p'_r,b,p'_{r+1})_{r\in\N}$ with label $b^{\omega} \not\in \L$.  Finally, $p'_l = f$ for some integer $l$ and
$$(p'_0, b, p'_1), \ldots (p'_{l-1}, b, p'_l=f=p_h),  \ldots (p_n, b, p_{n+1}),(p_{n+1}, a, p_{n+2})\ldots$$
is an accepting initial path with label $b^{l+n-h+1}a^{\omega} \not\in \calL$ and this is a contradiction.
\qed
\end{proof}

In a similar way as in Lemma \ref{lem:ab^*a(a+b)^{omega}}, one can prove the following.

\begin{lemma}
\label{lem:b^*ab^*a(a+b)^{omega}}
$b^*ab^*a(a+b)^{\omega} \not\in \FA(\mathbb{L}')$.
\end{lemma}

\begin{lemma}
\label{lem:(a+b)^{*}ba^{omega}}
$(a+b)^{*}ba^{\omega} \in \CFA(\mathbb{L}') \smallsetminus \DFA(\mathbb{L'})$.
\end{lemma}

\begin{proof}
Let $\L$ denote the language $(a+b)^{*}ba^{\omega}$. Consider the $\CFA$ $\aut'$ in Figure~\ref{fig:(a+b)^{*}ba^{omega}}. It is easy to see that $\L = \Lcond{\mathbb{L}'}{\aut'}$.

\begin{figure}
\begin{center}
\scalebox{1.0}{
\begin{tikzpicture}[semithick, shorten >=1pt, node distance=2cm, >=stealth',initial text=]

\node[state, initial,accepting] (A) {$q_0$};
\node[state] (B) [right of=A]{$q_1$};

\path
(A) edge[->, loop above] node[above] {$a,b$} (A)
    edge[->, bend left] node[above] {$b$} (B)
(B) edge[->, bend left] node[below] {$b$} (A)
    edge[->, loop above] node[above] {$a$} (B);
\end{tikzpicture}
}
\end{center}
\caption{$\CFA$ recognizing $(a+b)^{*}ba^{\omega}$ under $\mathbb{L}'$.}
\label{fig:(a+b)^{*}ba^{omega}}
\end{figure}
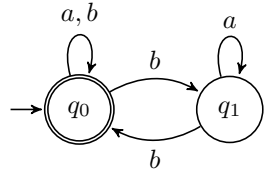

For a sake of argument, suppose that there exists a $\DFA$ $\aut = (\Sigma,Q,T,q_0,\calF)$ such that $\L = \Lcond{\mathbb{L}'}{\aut}$. By Lemma~\ref{1statof}, we can assume that $\calF=\{\{f\}\}$ with $f\in Q$. Let $n = \card{Q}$. Since $a^{n}ba^{\omega} \in \L$, there exists an accepting initial path
$$(p_0, a, p_1), \ldots (p_{n-1}, a, p_n),  (p_n, b, p_{n+1}),  (p_{n+1}, a, p_{n+2})\ldots$$ with label $a^{n}ba^{\omega}$. 
Since $Q$ is finite, $p_i = p_j$ for some $0 \leq i < j \leq n$ and 
$(p_0, a, p_1), \ldots (p_{j-1}, a, p_j),  (p_j=p_i, a, p_{i+1}),\ldots $
is an initial path with label $a^{\omega} \not\in \calL$. 
Then, $p_h = f$ for some integer $h$ with $i \leq h \leq j$.
Since the word  $b^{n+1}a^{\omega}$ also belongs to $\calL$, there exists an accepting initial path
$$(p'_0, b, p'_1), \ldots (p'_{n-1}, b, p'_n),  (p'_n, b, p'_{n+1}), (p_{n+1}, a, p_{n+2})\ldots$$ with label $b^{n+1}a^{\omega}$. Again, since $Q$ is finite, $p'_{i'} = p'_{j'}$ for some $1 \leq i' < j' \leq n + 1$ and the sequence $(p'_0, b, p'_1), \ldots (p_{j'-1}, b, p_{j'}),  (p_{j'}=p_{i'}, b, p_{i'+1}),\ldots $
is an initial path with label $b^{\omega} \not\in \calL$. This means that $p'_{k} = f$ for some integer $k$, $i' \leq k \leq j'$. Finally, 
$$(p'_0, b, p'_1), \ldots (p'_{k-1}, b, p'_k=f=p_h),  (p_{h}, a, p_{h+1}),\ldots (p_{j-1}, a, p_{j}=p_i)\ldots$$  
is a non-accepting initial path with label $b^{k}a^{\omega}$. Since $\aut$ is deterministic, there is no other path with label $b^{k}a^{\omega}$ and $b^{k}a^{\omega} \not\in \Lcond{\mathbb{L}'}{\A}$, and this is a contradiction.\qed
\end{proof}

\begin{proposition}
$\FA(\mathbb{L}') \subsetneq \FsR$.
\end{proposition}
\begin{proof}
For any $\FA$ $\A = (\Sigma,Q,T,q_0,\calF)$, by Lemma~\ref{1statof} we can assume that $\calF=\{\{f\}\}$. Define the $\FA$ $\A' = (\Sigma,Q,T,q_0,\set{Q \smallsetminus \set{f}})$. Then, $\Lcond{\mathbb{L}'}{\A} = \Lcond{(\inf, \subseteq)}{\A'}$ and, so, $\FA(\mathbb{L}')\subseteq \FA(\inf, \subseteq)$. Moreover, by the know fact $\FA(\inf, \subseteq)=\FsR$, we obtain that $\Lcond{(\inf, \subseteq)}{\A'}\in \FsR$. 
Lemma~\ref{lem:b^*ab^*a(a+b)^{omega}} gives the strict inclusion.\qed
\end{proof}

\begin{proposition}
$\DFA(\mathbb{L}')$ and $\CFA(\mathbb{L}')$ are incomparable.
\end{proposition}
\begin{proof}
It is an immediate consequence of  Lemmata~\ref{lem:ab^*a(a+b)^{omega}} and~\ref{lem:(a+b)^{*}ba^{omega}}.\qed
\end{proof}
\begin{proposition}
\label{L'incomparable}
The following statements are true:
\begin{enumerate}
\item
$\FA(\mathbb{L}')$ and $\GdR$ are incomparable,
\item
$\FA(\mathbb{L}')$ and $\GR$ are incomparable.
\end{enumerate}
\end{proposition}
\begin{proof}
By Lemma~\ref{lem:(a+b)^*a^{omega}}, $(a+b)^*a^{\omega}\in \CDFA(\mathbb{L}') \smallsetminus \GdR$ and, by Lemma~\ref{lem:b^*ab^*a(a+b)^{omega}}, $b^*ab^*a(a+b)^{\omega} \in \GR \smallsetminus \FA(\mathbb{L}')$. To conclude, recall that $\GR\subseteq \GdR$.\qed
\end{proof}

\begin{proposition}
$\CDFA(\mathbb{L}')$ and $\DFA(\fin, \sqcap)$ are incomparable.
\end{proposition}

\begin{proof}
By Proposition~\ref{L'incomparable} and by the known fact $G^{R}\subseteq \DFA(\fin, \sqcap)$, it follows that $\DFA(\fin, \sqcap) \not\subseteq \CDFA(\mathbb{L}')$. Furthermore, 
it has been shown in \cite{Litovsky1997} that $\CDFA(\mathbb{L}') \not\subseteq \DFA(\fin, \sqcap)$.\qed
\end{proof}

\section{Towards a characterization of $(\fin,=)$ and $(\fin,\subseteq)$.}

In this section we start studying the conditions $(\fin,=)$ and $(\fin,\subseteq)$.
Concerning $(\fin,=)$, Theorem \ref{th:(fin,=)=RAT} tells us
that,  in the non-deterministic case,  the class of recognized languages coincides with $\RAT$. In the deterministic case, either
it again coincides with $\RAT$ or it defines a completely new class 
(Proposition \ref{prop:partial-char-CDFA(fin,=)}).

\begin{proposition}
The following equality holds for $(\ninf,=)$:
\[
\CDFA(\ninf,=) = \DFA(\ninf,=) = \CFA(\ninf,=) = \FA(\ninf,=) = \RAT\enspace.
\]
\end{proposition}
\begin{proof}
For any $\FA$ $\A = (\Sigma,Q,T,q_0,\calF)$, let $\A' = (\Sigma,Q,T,q_0,
\{Q \smallsetminus F : F \in \calF\})$. Clearly, $\aut'$ is deterministic (resp. complete) if $\aut$ is deterministic (resp. complete).
It is not difficult to see that $\Lcond{(\ninf,=)}{\A} = \Lcond{(\inf, =)}{\A'}$ and $\Lcond{(\inf,=)}{\A} = \Lcond{(\ninf, =)}{\A'}$. Hence, it holds that $\FA(\ninf,=)=\FA(\inf,=)$, $\DFA(\ninf,=)=\DFA(\inf,=)$, $\CFA(\ninf,=)=\CFA(\inf,=)$, and $\CDFA(\ninf,=)=\CDFA(\inf,=)$. The known results on the language classes regarding $(\inf,=)$ conclude the proofs. \qed
\end{proof}

\begin{proposition}
The following equalities hold for $(\fin,\subseteq)$ and $(\fin,=)$:
\begin{align*}
\DFA(\fin,\subseteq) = \CDFA(\fin,\subseteq) \text{ and }\FA(\fin,\subseteq) = \CFA(\fin,\subseteq)\enspace,\\
\DFA(\fin,=) = \CDFA(\fin,=) \text{ and }\FA(\fin,=) = \CFA(\fin,=)\enspace.
\end{align*}
\end{proposition}
\begin{proof}
For any $\FA$ $\A = (\Sigma,Q,T,q_0,\calF)$, let $\A' = (\Sigma,Q \cup \set{\bot, \bot'},T',q_0,\calF)$ where
\begin{align*}
T'  =  T & \cup \set{(p,a,\bot) : p \in Q, a \in \Sigma, \forall q \in Q, (p,a,q) \not\in T} \cup \set{(\bot, a, \bot') : a \in \Sigma}\\
&\cup \set{(\bot', a, \bot') : a \in \Sigma}
\end{align*}
The $\FA$ $\A'$ is complete. Moreover, $\A'$ is a $\DFA$ if and only if $\A$ is a $\DFA$. Furthermore, under both the conditions $(\fin,\subseteq)$ and $(\fin,=)$, every accepting path in $\aut$ is still an accepting path in $\aut'$, and if $p$ is an initial path in $\aut'$ which is not a path in $\aut$, then $\bot \in \fin(p)$. Since $\forall F \in \calF, \bot \not\in F$, the path $p$ is non accepting in $\aut'$. Therefore,  $\Lcond{(\fin, \subseteq)}{\aut} = \Lcond{(\fin, \subseteq)}{\aut'}$ and $\Lcond{(\fin, =)}{\aut} = \Lcond{(\fin, =)}{\aut'}$ and this concludes   the proof.
\end{proof}

\begin{proposition}[Staiger \cite{staiger1997}]\mbox{}\\
$\CDFA(\fin, \subseteq) \subseteq \CDFA(\fin, =)$ and $\CFA(\fin, \subseteq) \subseteq \CFA(\fin, =)$.
\end{proposition}

\begin{proof}
For any $\CDFA$ (resp. $\CFA$) $\A = (\Sigma, Q, T, q_0, \calF)$, define the $\CDFA$ (resp. $\CFA$) $\A' = (\Sigma, Q, T, q_0, \bigcup_{F \in \calF} \{\powerset{F}\})$. Then, it follows that $\Lcond{(\fin, \subseteq)}{\A} = \Lcond{(\fin, =)}{\A'}$ and this concludes the proof.
\qed
\end{proof}

\begin{proposition}[Staiger \cite{staiger1997}]\mbox{}\\
$\FA(\fin, \sqcap) \subseteq \FA(\fin,=)$ and $\DFA(\fin, \sqcap) \subseteq \DFA(\fin,=)$.
\end{proposition}

\begin{proof}
For any $\FA$ $\A = (\Sigma,Q,T,q_0,\calF)$, let $\A' = (\Sigma,Q,T,q_0,\calF')$ where $\calF' = \{F \in \powerset{Q} : \exists X \in \calF, X \cap F \neq \emptyset\}$. Then, $\Lcond{(\fin, \sqcap)}{\A} = \Lcond{(\fin, =)}{\A'}$.
It is clear that $\A'$ is a $\DFA$ if $\A$ is a $\DFA$, and this concludes the proof.\qed
\end{proof}

\begin{lemma}
\label{lem_rat_fin_eq}
$\RAT \subseteq \FA(\fin, =)$.
\end{lemma}
\begin{proof}
We are going to show that $\FA(\inf,\sqcap)\subseteq \FA(\fin, =)$, \ie, for any $\FA$ $\A = (\Sigma,Q,T,q_0,\calF)$ there exists a $\FA$ $\aut'$ such that $\Lcond{(\inf,\sqcap)}{\A}=\Lcond{(\fin,=)}{\A'}$. The known fact that $\RAT=\FA(\inf,\sqcap)$ concludes the proof.

Let $\A' = (\Sigma,Q \cup Q \times Q,T',q_0,\calF')$ where
$$T' = T \cup \set{(p,a,(q,p)) : (p,a,q) \in T} \cup \set{((p_1, p_2), a, q) : (p_1,a,q) \in T, p_2 \in Q}$$
and
$$\calF' = \set{F \smallsetminus \set{p_2} \cup \set{(p_1,p_2)} : p_1 \in Q, F \in \powerset{Q}, \exists X \in \calF, p_2 \in X}\enspace.$$

We prove that $\Lcond{(\inf,\sqcap)}{\A}\subseteq\Lcond{(\fin,=)}{\A'}$. Let $x \in \Lcond{(\inf,\sqcap)}{\A}$.  There exists a path $p = (p_i, x_i, p_{i+1})_{i \in \N}$ in $\aut$, a state $q \in Q$ and a set $F \in \calF$ such that $q \in F$ and $q = p_i$ for infinitely many $i \in \N$. Let $n > 0$ be such that $p_{n} = q$ and let $p' = (p'_i, x_i, p'_{i+1})_{i \in \N}$ be the initial path in $\A'$ defined by $\forall i \neq n + 1, p'_i = p_i$ and $p'_{n+1} = (p_{n+1},q)$. As $q \not\in \fin(p')$, $\fin(p') = (\fin(p') \cap Q) \smallsetminus \set{q} \cup \set{(p_{n+1}, q)} \in \calF'$. Hence, $x \in \Lcond{(\fin,=)}{\A'}$.

We now show that $\Lcond{(\fin,=)}{\A'}\subseteq\Lcond{(\inf,\sqcap)}{\A}$. Let $x \in \Lcond{(\fin,=)}{\A'}$. There exists a path $p = (p_i, x_i, p_{i+1})_{i \in \N}$ in $\A'$, two states $q_1, q_2\in Q$ and a set $F \in \powerset{Q}$ such that $\exists X \in \calF$ with $q_2 \in X$ and $\fin(p) = F \smallsetminus \set{q_2} \cup \set{(q_1,q_2)}$. Let $p' = (p'_i, x_i, p'_{i+1})_{i \in \N}$ be the initial path in \A defined by $\forall i \in \N, p'_i = p_i$ if $p_i \in Q$,  $p'_i = a_i$ with $p_i = (a_i, b_i) \in Q \times Q$, otherwise. As $(q_1,q_2) \in \fin(p)$, $q_2 \in \run(p)$ (because $q_2$ is the only possible predecessor of $(q_1,q_2)$) but $q_2 \not\in \fin(p)$, then $q_2 \in \inf(p) \subseteq \inf(p')$. Hence, $x \in \Lcond{(\inf,\sqcap)}{\A}$.\qed
\end{proof}

\begin{theorem}\label{th:(fin,=)=RAT}
$\FA(\fin,=) = \RAT$.
\end{theorem}
\begin{proof}
Combine Lemma~\ref{lem_rat_fin_eq} and Proposition~\ref{prop:all_rat}.\qed
\end{proof}

The Proposition~\ref{prop:partial-char-CDFA(fin,=)} shows that in the deterministic case,  either $(fin,=)$ induces $\RAT$ or it defines a new class outside the Borel hierarchy.

\begin{proposition}\label{prop:partial-char-CDFA(fin,=)}
$a(a^*b)^{\omega} + b(a+b)^*a^{\omega} \in \CDFA(\fin,=) \smallsetminus (\FsR \cup \GdR)$.
\end{proposition}

\begin{proof}
In \cite{landweber1969}, it is proved that $\L=a(a^*b)^{\omega} + b(a+b)^*a^{\omega}  \not\in \FsR \cup \GdR$. To conclude,  it is enough to remark that
$\L = \Lcond{(\fin,=)}{\A}$ for the $\CDFA$
$$\A = (\set{a,b}, \set{q_0, q_1,q_2,q_3,q_4,q_5}, T, q_0, \set{\emptyset, \set{q_2}, \set{q_3,q_4}})\enspace,$$
where the set of transitions is given in Figure~\ref{fig_CDFA_fin_eq}.

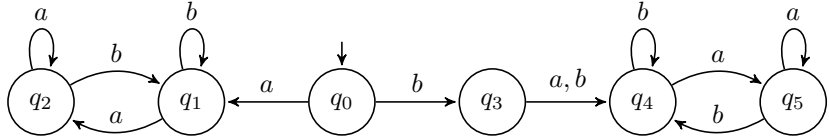
\begin{figure}[htb]
\begin{center}
\scalebox{1.0}{
\begin{tikzpicture}[semithick, shorten >=1pt, node distance=2cm, >=stealth',initial text=]

\node[state] (A) {$q_0$};
\node[state] (B) [left of=A]{$q_1$};
\node[state] (C) [left of=B]{$q_2$};
\node[state] (D) [right of=A]{$q_3$};
\node[state] (E) [right of=D]{$q_4$};
\node[state] (F) [right of=E]{$q_5$};

\draw[->]   (0,0.8)-- (A);
\path
(A) edge[->] node[above] {$a$} (B)
    edge[->] node[above] {$b$} (D)
(B) edge[->, bend left] node[above] {$a$} (C)
    edge[->, loop above] node[above] {$b$} (B)
(C) edge[->, bend left] node[above] {$b$} (B)
    edge[->, loop above] node[above] {$a$} (C)
(D) edge[->] node[above] {$a,b$} (E)
(E) edge[->, bend left] node[above] {$a$} (F)
    edge[->, loop above] node[above] {$b$} (E)
(F) edge[->, bend left] node[above] {$b$} (E)
    edge[->, loop above] node[above] {$a$} (F);

\end{tikzpicture}
}
\end{center}
\caption{A \CDFA recognizing $a(a^*b)^{\omega} + b(a+b)^*a^{\omega}$ under $(\fin,=)$.}
\label{fig_CDFA_fin_eq}
\end{figure}

Let $p = (p_i,a_i,p_{i+1})_{i \in \N}$ be an accepting path in \A. If $a_0 = b$, then $p_1 = q_3$ and $p_2 = q_4$. As $q_3$ is not reachable from $q_4$ and $p$ is accepting, $\fin_\A(p) = \set{q_3,q_4}$ and $q_4$ is visited finitely often, then the label of $p$ contains only finitely many $b$'s.

If $a_0 = a$, then $p_1 = q_1$. As $q_3$ is not reachable from $q_1$ and $p$ is accepting, $\fin_\A(p) = \emptyset$ or $\fin_\A(p) = \set{q_2}$. In both cases, $q_1$ is not visited finitely many times and as it is visited once, it is visited infinitely often. Then the label of $p$ contains infinitely many $b$'s.

Conversely, it is easy to see that a path $p$ is accepting when
\begin{itemize}
\item
its label starts by a $b$ and contains finitely many $b$'s ($\fin_\A(p) = \set{q_3,q_4}$)\enspace,
\item
its label is $ab^{\omega}$ or it starts by a $a$ and contains infinitely many $a$'s and $b$'s ($\fin_\A(p) = \emptyset$)\enspace,
\item
its label starts by a $a$ and contains infinitely many $b$'s but only finitely many $a$'s ($\fin_\A(p) = \set{q_2}$)\enspace.
\end{itemize} \qed
\end{proof}

\begin{proposition}
\label{lem_fin_subset_det_Fsigma}
$\DFA(\fin, \subseteq) \subseteq \GdR$.
\end{proposition}

\begin{proof}
Let $\A = (\Sigma,Q,T,q_0,\calF)$ be a \DFA. For any $S \subseteq Q$, let $\A_S$ be the \DFA $(\Sigma,Q,T,q_0, \set{S})$. Let $\calL$ denote the language
\[
\bigcup_{S, S \subseteq Q, \exists F \in \calF, S \smallsetminus S' \subseteq F}
\left(\Lcond{(\run,\subseteq)}{\A_S} \cap \bigcap_{q \in S'} \Lcond{(\inf,\sqcap)}{\A_{\set{q}}}\right) \enspace,
\]
then $\Lcond{(\fin,\subseteq)}{\A}  = \calL$.

First, we prove that $\Lcond{(\fin,\subseteq)}{\A}  \subseteq \calL$. Let $x \in \Lcond{(\fin,\subseteq)}{\A}$, there exists an accepting path in $\A$ under $(\fin,\subseteq)$ labeled by $x$, \ie, there exists $F \in \calF$ such that $\fin_{\A}(p) = \run_{\A}(p) \smallsetminus \inf_{\A}(p) \subseteq F$. For this path, take $S = \run_{\A}(p)$ and $S' = \inf_{\A}(p)$, we obtain
\[
x \in \Lcond{(\run,\subseteq)}{\A_S} \cap \bigcap_{q \in S'} \Lcond{(\inf,\sqcap)}{\A_{\set{q}}} \subseteq \calL \enspace.
\]
Conversely, we prove that $\calL \subseteq \Lcond{(\fin,\subseteq)}{\A}$. Let $x \in \calL$, by determinism, there exists a path $p$ in $\A$ labeled by $x$ such that there exist $S, S' \subseteq Q$, $F \in \calF$ with $S \smallsetminus S' \subseteq F$ such that $p$ is accepting for $\A_S$ under $(\run,\subseteq)$ and for $A_{\set{q}}$ under $(\inf, \sqcap)$ for all $q \in S'$. The path $p$ verifies $\run_{\A}(p) \subseteq S$, $S' \subseteq \inf_{\A}(p)$ and then $\fin_{\A}(p) \subseteq S \smallsetminus S' \subseteq F$. Finally, $p$ is accepting for $\A$ under $(\fin,\subseteq)$ and $x \in \Lcond{(\fin,\subseteq)}{\A}$.

 For all $S \subseteq Q$, $\Lcond{(\run,\subseteq)}{\A_S} \in \FR \subseteq \GdR$ and $\Lcond{(\inf,\sqcap)}{\A_S} \in \GdR$. As $\GdR$ is stable by finite intersection and union, $\Lcond{(\fin,\subseteq)}{\A} \in \GdR$.

\qed
\end{proof}
\begin{figure}[htb]
\begin{center}
\scalebox{0.9}{
\begin{tikzpicture}[semithick, shorten >=1pt, >=stealth']
\newcommand{\esph}{4.5}
\newcommand{\espv}{1.5}
\newcommand{\myxshift}{-20}

\tikzstyle{vertex}=[draw, shape=rectangle, font=\tiny]

\node[vertex,xshift=\myxshift] (A) at (0*\esph,0.5*\espv)
{
    $\begin{array}{c}
    \boldsymbol{\RAT} \\
    \FA(\inf,\sqcap)~\CFA(\inf,\sqcap)\\
    \FA(\inf,=)~\DFA(\inf,=)~\CFA(\inf,=)~\CDFA(\inf,=)\\
    \FA(\ninf, \subseteq)~\CFA(\ninf,\subseteq) \\
    \FA(\ninf,=)~\DFA(\ninf,=)~\CFA(\ninf,=)~\CDFA(\ninf,=)\\
    \FA(\fin,=)~\CFA(\fin,=)
    \end{array}$
};
\node[vertex,xshift=\myxshift-28] (B) at (.7*\esph,-1*\espv)
{
    $\begin{array}{c}
    \boldsymbol{\FsR} \\
    \FA(\run,\sqcap)\\
    \FA(\run,=)~\CFA(\run,=)\\
    \FA(\inf,\subseteq)~\DFA(\inf,\subseteq)~\CFA(\inf,\subseteq)~\CDFA(\inf,\subseteq) \\
    \FA(\fin,\sqcap)\\
    \FA(\mathbb{A})
    \end{array}$
};
\node[vertex] (C) at (-1*\esph,-\espv) 
{
    $\begin{array}{c}
    \boldsymbol{\GdR} \\
    \DFA(\inf,\sqcap)~\CDFA(\inf,\sqcap) \\
    \DFA(\ninf, \subseteq)~\CDFA(\ninf,\subseteq)
    \end{array}$
};
\node[vertex,xshift=\myxshift] (D) at (0*\esph,-2*\espv)
{
    $\begin{array}{c}
    \boldsymbol{\FsR \cap \GdR} \\
    \DFA(\run,=)~\CDFA(\run,=)
    \end{array}$
};
\node[vertex] (E) at (-1*\esph,-5*\espv)
{
    $\begin{array}{c}
    \boldsymbol{\FR}\\
    \FA(\run,\subseteq)~\DFA(\run, \subseteq)~\CFA(\run,\subseteq)~\CDFA(\run, \subseteq)\\
    \FA(\mathbb{A}')~\DFA(\mathbb{A}')~\CFA(\mathbb{A}')~\CDFA(\mathbb{A}')
    \end{array}$
};
\node[vertex,xshift=\myxshift-36] (F) at (1*\esph,-5*\espv)
{
     $\begin{array}{c}
   \boldsymbol{\GR}\\
    \CFA(\run,\sqcap)~\CDFA(\run,\sqcap)\\
    \CFA(\mathbb{A})~\CDFA(\mathbb{A})
    \end{array}$
};
\node[vertex,xshift=\myxshift] (G) at (0*\esph,-6*\espv)
{
    $\boldsymbol{\FR \cap \GR}$
};
\node[vertex,xshift=\myxshift] (H) at (0*\esph,-4*\espv)
{
    $\begin{array}{c}
    \boldsymbol{\FsR~\Delta~\GdR} \\
    \DFA(\run,\sqcap) \\
    \DFA(\mathbb{A})
    \end{array}$
};
\node[vertex,xshift=\myxshift-36] (I) at (1*\esph,-4*\espv)
{
    $\CDFA(\fin,\sqcap)$
};
\node[vertex,xshift=\myxshift-8] (J) at (0.5*\esph,-3*\espv)
{
    $\DFA(\fin,\sqcap)$
};
\node[vertex,xshift=\myxshift-68] (K) at (1.5*\esph,-3*\espv)
{
    $\CFA(\fin,\sqcap)$
};
\node[vertex] (L) at (-1*\esph,-4*\espv)
{
    $\CDFA(\ninf,\sqcap)$
};
\node[vertex] (M) at (-0.5*\esph,-3*\espv)
{
    $\CFA(\ninf,\sqcap)$
};
\node[vertex] (N) at (-1.5*\esph,-3*\espv)
{
    $\DFA(\ninf,\sqcap)$
};
\node[vertex] (O) at (-1*\esph,-2*\espv)
{
    $\FA(\ninf,\sqcap)$
};

\draw[->] (B) -- (A);
\draw[->] (C) -- (A);
\draw[->] (D) -- (B);
\draw[->] (D) -- (C);
\draw[->] (E) -- (H);
\draw[->] (F) -- (H);
\draw[->] (H) -- (D);
\draw[->] (G) -- (E);
\draw[->] (G) -- (F);
\draw[->] (F) -- (I);
\draw[->] (H) -- (J);
\draw[->] (I) -- (J);
\draw[->] (I) -- (K);
\draw[->] (J) -- (B);
\draw[->] (K) -- (B);
\draw[->] (E) -- (L);
\draw[->] (L) -- (M);
\draw[->] (L) -- (N);
\draw[->] (M) -- (O);
\draw[->] (N) -- (O);
\draw[->] (O) -- (B);

\end{tikzpicture}
}
\end{center}
\caption{The completion of Figure~\ref{fig:hierarchy-before} with the results in the paper.
Classes of the Borel hierarchy are typeset in bold. Arrows mean strict inclusion. Classes in the same box coincide.}
\label{fig:hierarchy-after}
\end{figure}

\section{Conclusions}

In this paper we have studied the expressivity power of acceptance conditions for finite automata.
Three new classes have been fully characterized. For a fourth one, partial results are given.
In particular, $(\ninf,\sqcap)$ provides four distinct new classes of languages (see the diamond
in the left part of Figure \ref{fig:hierarchy-after}), all other acceptance conditions considered tend
to give (classes of) languages populating known classes.

In literature, other well-known acceptance conditions exists for example Rabin, Strett or Parity 
conditions. These last ones have not been taken into account in the present paper since it is known
that they are equivalent to Muller's condition. 

Several research directions should be further explored but at least two seems the more promising ones.
First, to complete the characterization of $(\fin,=)$. Moreover, the exact position of $(\fin,\subseteq)$ 
in the hierarchy given so far is still under investigation.

Second, to study the closure properties of the the new classes of languages introduced in the paper  and 
verify if they cram the known classes or if they add new elements to Figure \ref{fig:hierarchy-after}.

\bibliographystyle{plain}
\bibliography{references}

\end{document}